\documentclass[aps,pre,twocolumn,groupedaddress]{revtex4-1}

\usepackage{euscript,amsmath,amssymb,amsfonts,graphicx,bm,amsthm,mathtools}

  \usepackage{paralist}
  \usepackage{epstopdf}
  \usepackage{graphics} 
 \usepackage[latin1]{inputenc}
\usepackage[caption=false]{subfig}
\usepackage{soul}

\usepackage{latexsym,amssymb,enumerate,amsmath,amsthm} 
  \usepackage{graphicx} 
  \usepackage{tikz}
  \usepackage{hyperref}
\usepackage{latexsym,amssymb,enumerate,amsmath} 
\usepackage{pgfplots}
\usepackage{soul,xcolor}

\newcommand{\dd}{\textup{d}}

\def\E{\mathbb{E}}
\def\P{\mathbb{P}}
\def\R{\mathbb{R}}

\def\c{\mathbf{c}}

\def\p{\mathbf{p}}

\def\r{\mathbf{r}}
\def\Dalt{\frac{\partial^{1-\alpha}}{\partial t^{1-\alpha}}}
\def\D{\prescript{}{0}D_{t}^{1-\alpha}}
\def\DD{\mathcal{D}_{t}}
\def\L{\mathcal{L}_{x}}
\def\domain{V}
\newcommand{\Markov}[2]{\underset{#1}{\overset{#2}{\rightleftharpoons}}}

\newtheorem{theorem}{Theorem}

\theoremstyle{plain}
\theoremstyle{remark}

\begin{document}


\title[]{Anomalous reaction-diffusion equations for linear reactions}


\author{Sean D. Lawley}
\email[]{lawley@math.utah.edu}
\affiliation{University of Utah, Department of Mathematics, Salt Lake City, UT 84112 USA}


\date{\today}

\begin{abstract}
Deriving evolution equations accounting for both anomalous diffusion and reactions is notoriously difficult, even in the
simplest cases. In contrast to normal diffusion, reaction kinetics cannot be incorporated into evolution equations modeling subdiffusion by merely adding reaction terms to the equations describing spatial movement. A series of previous works derived fractional reaction-diffusion equations for the spatiotemporal evolution of particles undergoing subdiffusion in one space dimension with linear reactions between a finite number of discrete states. In this paper, we first give a short and elementary proof of these previous results. We then show how this argument gives the evolution equations for more general cases, including subdiffusion following any fractional Fokker-Planck equation in an arbitrary $d$-dimensional spatial domain with time-dependent reactions between infinitely many discrete states. In contrast to previous works which employed a variety of technical mathematical methods, our analysis reveals that the evolution equations follow from (i) the probabilistic independence of the stochastic spatial and discrete processes describing a single particle and (ii) the linearity of the integro-differential operators describing spatial movement. We also apply our results to systems combining reactions with superdiffusion.
\end{abstract}

\pacs{}

\maketitle

\section{\label{intro}Introduction}

The signature of a normal diffusive process is that the mean-squared displacement grows linearly in time. That is, if $X(t)$ denotes the one-dimensional position of the diffusive particle at time $t\ge0$, then
\begin{align}\label{normal}
\E\big[\big(X(t)-X(0)\big)^{2}\big]
\propto t,
\end{align}
where $\E$ denotes expected value. However, the mean-squared displacement in complex systems often deviates from the linear behavior in \eqref{normal} and instead grows as a power law,
\begin{align}\label{alpha}
\E\big[\big(X(t)-X(0)\big)^{2}\big]
\propto t^{\alpha},\quad \alpha>0,
\end{align}
in a phenomenon called anomalous diffusion if $\alpha\neq1$. Subdiffusion is defined by \eqref{alpha} with $\alpha<1$ and has been observed in various systems, including charge transport in amorphous semiconductors \cite{scher1975}, subsurface hydrology \cite{berkowitz2002}, and the transport of a bead through a polymer network \cite{amblard1996}. In addition, subdiffusive motion is ubiquitous in cell biology, where it is believed to result from macromolecular crowding \cite{golding2006, hofling2013}. Superdiffusion is defined by \eqref{alpha} with $\alpha>1$ and has been observed in animal movement \cite{klafter2005} and in active transport inside cells \cite{caspi2000}. 

Three common mathematical models for subdiffusion are the continuous-time random walk model, fractional Brownian motion, and random walks on fractal and disordered systems \cite{hofling2013}. In a continuum limit, the { standard} continuous-time random walk model {with independent jump length and waiting time distributions} yields the fractional diffusion equation \cite{metzler2000},
\begin{align}\label{fde}
\frac{\partial}{\partial t}c(x,t)
=\D K_{\alpha}\frac{\partial^{2}}{\partial x^{2}} c(x,t),\quad x\in\R,\,t>0,
\end{align}
for the subdiffusive chemical concentration $c(x,t)$ at position $x$ at time $t$. In \eqref{fde}, the parameter $K_{\alpha}>0$ is the generalized diffusivity (with dimensions $(\text{length})^{2}(\text{time})^{-\alpha}$) and $\D$ is the Riemann-Liouville fractional derivative \cite{samko1993}, defined by
\begin{align}\label{rl}
\D f(t)
=\frac{1}{\Gamma(\alpha)}\frac{\dd}{\dd t}\int_{0}^{t}\frac{f(s)}{(t-s)^{1-\alpha}}\,\dd s,
\end{align}
where $\Gamma(\alpha)$  is the Gamma function. Note that $\D$ is sometimes denoted by $\Dalt$. {As a technical aside, the operator appearing in the derivation of \eqref{fde} is actually the Gr{\"u}nwald-Letnikov derivative, but this operator is equivalent to \eqref{rl} for sufficiently smooth functions \cite{podlubny1998}.}

Generalizing \eqref{fde}, fractional Fokker-Planck equations model the spatiotemporal evolution of subdiffusive molecules under the influence of an external force \cite{metzler1999}. A fractional Fokker-Planck equation takes the form
\begin{align}\label{ffpec}
\frac{\partial}{\partial t}c(x,t)
=\D \L c(x,t),\quad x\in\domain\subseteq\R^{d},\,t>0,
\end{align}
where $\domain\subseteq\R^{d}$ is a $d$-dimensional spatial domain and $\L$ is the forward Fokker-Planck operator,
\begin{align}\label{fpo}
\begin{split}
\L f(x)
&:=-\sum_{i=1}^{d}\frac{\partial}{\partial x_{i}}\big[\mu_{i}(x)f(x)\big]\\
&+\frac{1}{2}\sum_{i=1}^{d}\sum_{k=1}^{d}\frac{\partial^{2}}{\partial x_{i}\partial x_{k}}
\Big[\big({\sigma}(x){\sigma}(x)^{\top}\big)_{i,k}f(x)\Big], 
\end{split}
\end{align}
where $\mu(x):\overline{\domain}\mapsto\R^{d}$ is the external force (drift) vector and $
\sigma(x):\overline{\domain}\mapsto\R^{d\times m}$ describes the space-dependence and anisotropy in the diffusivity. Of course, if $\alpha=1$, then $\D$ is the identity operator and \eqref{ffpec} reduces to the familiar equation of integer order,
\begin{align}\label{fpec}
\frac{\partial}{\partial t}c(x,t)
=\L c(x,t). 
\end{align}

A fundamental and now longstanding question is how to model reaction kinetics for subdiffusive molecules (see the review \cite{nepomnyashchy2016} and \cite{hornung2005, gafiychuk2008, boon2012, angstmann2013, kosztolowicz2013, hansen2015, straka2015, dossantos2019, zhang2019, li2019}). In the classical case of normal diffusion, reaction terms can simply be added to the evolution equations describing spatial movement. More precisely, consider the vector of $n$ chemical concentrations,
\begin{align}\label{c}
\c(x,t)=(\c_{i}(x,t))_{i=1}^{n}\in\R^{n},
\end{align}
where $\c_{i}(x,t)$ denotes the concentration of species $i$ at position $x\in\R^{d}$ at time $t\ge0$. In the absence of spatial movement, suppose the concentrations obey the mean-field reaction equations
\begin{align}\label{reactions}
\frac{\partial}{\partial t}\c
=\mathbf{f}(\c),
\end{align}
where $\mathbf{f}:\R^{n}\mapsto\R^{n}$. In the case of normal diffusion where each chemical species moves by \eqref{fpec}, one incorporates the reaction kinetics \eqref{reactions} into spatiotemporal evolution equations by the simple addition of $\mathbf{f}(\c)$ to the righthand side,
\begin{align}\label{simpleform}
\frac{\partial}{\partial t}\c
=\L\c+\mathbf{f}(\c).
\end{align}
However, this simple procedure fails for subdiffusion. Indeed, it was shown that the following attempt to combine subdiffusion with degradation at rate $\lambda>0$, 
\begin{align*}
\frac{\partial}{\partial t}c
=\D K_{\alpha}\frac{\partial^{2}}{\partial x^{2}} c-\lambda c,\quad x\in\R,\,t>0,
\end{align*}
leads to an unphysical negative concentration, $c(x,t)<0$ \cite{henry2006}.

In a series of important works \cite{sokolov2006, henry2006, schmidt2007, langlands2008}, evolution equations were derived for certain subdiffusive processes with linear reactions. In \cite{sokolov2006}, the equations were derived for pure subdiffusion in $\R$ with an irreversible reaction between $n=2$ chemical species. Equivalent equations were then derived in \cite{henry2006} and \cite{schmidt2007} using different formalisms. These results were generalized in \cite{langlands2008} to allow reversible reactions between any finite number $n$ of chemical species. In particular, in the case that (i) the reactions in \eqref{reactions} are linear,
\begin{align*}
\mathbf{f}(\c)
=R\c,
\end{align*}
where $R\in\R^{n\times n}$ is a constant reaction rate matrix, and (ii) each chemical species moves by the one-dimensional fractional diffusion equation in \eqref{fde}, it was found that \cite{langlands2008}
\begin{align}\label{langlands}
\frac{\partial}{\partial t}\c
=K_{\alpha}e^{{{R}} t}\D\Big[e^{-{{R}} t}\frac{\partial^{2}}{\partial x^{2}}\c\Big]
+{{R}} \c,\quad x\in\R,
\end{align}
where $e^{Rt}$ is the matrix exponential. In contrast to the simple form in \eqref{simpleform} with decoupled reaction and movement terms, notice that the reactions modify the movement term in \eqref{langlands}.  Interestingly, for the case of L\'{e}vy flights with an irreversible reaction between $n=2$ species, it was shown in \cite{schmidt2007} that the reaction-superdiffusion equations have the usual decoupling of reaction and movement terms. The results in \cite{sokolov2006, schmidt2007, henry2006, langlands2008} were derived using continuous-time random walks and Fourier-Laplace transform theory.

In this paper, we first give a short and elementary proof of \eqref{langlands}. We then show how this argument gives the evolution equations for more general cases, including subdiffusion following any fractional Fokker-Planck equation in an arbitrary $d$-dimensional spatial domain with time-dependent reactions between infinitely many discrete states. This analysis reveals that the evolution equations follow from (i) the probabilistic independence of the stochastic spatial and discrete processes describing a single particle and (ii) the linearity of the integro-differential operators describing spatial movement. In addition, under mild assumptions on initial and boundary conditions, the evolution equations imply that the spatial and discrete processes are independent. That is, under some mild conditions, the evolution equations hold if and only if the spatial and discrete processes are independent.

The rest of the paper is organized as follows. In section~\ref{simplifiedsetting}, we give a simple argument that yields \eqref{langlands}. In section~\ref{moregeneralsetting}, we generalize this argument to yield the evolution equations describing more complicated spatial and discrete processes. In section~\ref{examples}, we apply this more general result to some examples. We conclude by discussing our results and highlighting future directions.

\section{Simplified setting\label{simplifiedsetting}}

We first consider a setup that is equivalent to the main problem considered in \cite{langlands2008}. Assume $\{J(t)\}_{t\ge0}$ is a continuous-time Markov jump process on the finite state space $\{1,\dots,n\}$. Suppose the matrix $R\in\R^{n\times n}$ contains the transition rates, meaning the distribution of $J(t)$ satisfies the linear ordinary differential equation,
\begin{align}\label{ode}
\frac{\dd}{\dd t}\r
={{R}} \r,
\end{align}
where $\r(t)$ is the vector of probabilities,
\begin{align}\label{vp}
\r(t)
=(\r_{i}(t))_{i=1}^{n}
:=\big(\P(J(t)=i)\big)_{i=1}^{n}
\in\R^{n}.
\end{align}
Of course, the solution to \eqref{ode} is the matrix exponential,
\begin{align}\label{me}
\r(t)
=e^{{{R}} t}\r(0),\quad t\ge0.
\end{align}
In the language of Markov chain theory, $R$ is the forward operator and the transpose $R^{\top}$ is the backward operator (i.e.\ $R^{\top}$ is the infinitesimal generator \cite{norris1998}).

Assume that $\{X(t)\}_{t\ge0}$ is a one-dimensional subdiffusive process taking values in $\R$. Let $q(x,t)$ denote the probability density that $X(t)=x$ and assume that it satisfies the fractional diffusion equation,
\begin{align}\label{ffp}
\frac{\partial}{\partial t}q
=\DD \L q,\quad x\in\R,\,t>0,
\end{align}
where
\begin{align}\label{sub0}
\DD=\D,\quad \alpha\in(0,1),
\end{align}
is the fractional derivative of Riemann-Liouville type given in \eqref{rl}, and
\begin{align}\label{sub1}
\L
={K_{\alpha}}\frac{\partial^{2}}{\partial x^{2}},
\end{align}
is the one-dimensional Laplacian with generalized diffusivity ${K_{\alpha}>0}$. We use the subscripts in \eqref{sub0} and \eqref{sub1} to emphasize that $\DD$ acts only on the time variable $t$ and $\L$ acts only on the space variable $x$.

Langlands et al.\ \cite{langlands2008} developed a mesoscopic continuous-time random walk argument to derive the following system of fractional reaction-diffusion equations,
\begin{align}\label{langlandsp}
\frac{\partial}{\partial t}\p
=e^{{{R}} t}\DD\Big[e^{-{{R}} t}\L \p\Big]
+{{R}} \p,\quad x\in\R,\,t>0,
\end{align}
for the joint density
$\p(x,t)=(\p_{i}(x,t))_{i=1}^{n}$, where
\begin{align*}
\p_{i}(x,t)\,\dd x
=\P(X(t)=x,\,J(t)=i). 
\end{align*}
The derivation in \cite{langlands2008} implicitly assumed that $X$ and $J$ are independent processes.

We now prove that \eqref{langlandsp} follows immediately from \eqref{ode}, \eqref{ffp}, the independence of $X$ and $J$, and the linearity of $\DD$ and $\L$. Note first that independence ensures that the joint probability distribution is the product of the individual distributions,
\begin{align}\label{pqr}
\begin{split}
\p_{i}(x,t)\,\dd x
&=\P(X(t)=x,\, J(t)=i)\\
&=\P(X(t)=x)\P(J(t)=i)\\
&=q(x,t)\r_{i}(t)\,\dd x.
\end{split}
\end{align}
Therefore, differentiating $\p(x,t)=q(x,t)\r(t)$ with respect to time and using \eqref{ode} and \eqref{ffp} yields
\begin{align}\label{calc}
\begin{split}
\frac{\partial}{\partial t}\p(x,t)
&=\r(t)\DD\Big[\L q(x,t)\Big]+q(x,t){{R}} \r(t)\\
&=\r(t)\DD\Big[\L q(x,t)\Big]+{{R}} \p(x,t).
\end{split}
\end{align}
Using \eqref{me}, the first term in the righthand side of \eqref{calc} becomes
\begin{align}\label{calc2}
\begin{split}
\r(t)\DD\Big[\L q(x,t)\Big]
&=e^{{{R}} t}\r(0)\DD\Big[\L q(x,t)\Big]\\
&=e^{{{R}} t}\DD\Big[\r(0)\L q(x,t)\Big]\\
&=e^{{{R}} t}\DD\Big[e^{-{{R}} t}\r(t)\L q(x,t)\Big]\\
&=e^{{{R}} t}\DD\Big[e^{-{{R}} t}\L \p(x,t)\Big].
\end{split}
\end{align}
Combining \eqref{calc} and \eqref{calc2} yields \eqref{langlandsp}.

\section{More general setting\label{moregeneralsetting}}

It is easy to see that the calculation in \eqref{pqr}-\eqref{calc2} and the resulting evolution equation in \eqref{langlandsp} holds in much greater generality. First, the spatial domain need not be $\R$, and we will instead take it to be any $d$-dimensional open set $\domain\subseteq\R^{d}$ with $d\ge1$. Second, the operator $\L$ need not be the Laplacian and the operator $\DD$ need not be the Riemann-Liouville fractional derivative. Instead, we will take $\L$ to be any linear operator acting on functions of space $x\in V\subseteq\R^{d}$ and $\DD$ to be any linear operator acting on functions of time $t\in[0,\infty)$. That is, if $\varphi(t),\psi(t)$ are real-valued functions of time $t\in[0,\infty)$ {in the domain of $\DD$} and $f(x),g(x)$ are real-valued functions of space $x\in\domain\subseteq\R^{d}$ {in the domain of $\L$}, then we assume
\begin{align}\label{linear}
\begin{split}
\L(\varphi f+\psi g)
&=\varphi \L f+\psi \L g,\\
\DD(\varphi f+\psi g)
&=f \DD \varphi+g \DD \psi.
\end{split}
\end{align}
Third, the jump process $J(t)$ need not have constant jump rates or a finite state space. We summarize this in the following theorem. {Equation~\eqref{evo} in Theorem~\ref{thmind} and its proof is the main result of this paper.}

\begin{theorem}\label{thmind}
Assume $\{J(t)\}_{t\ge0}$ is a stochastic process on the possibly infinite, countable state space, $\{1,2,\dots,n\}$, where
\begin{align*}
n\in\mathbb{N}\cup\{\infty\}.
\end{align*}
Suppose the distribution,
\begin{align*}
\r(t):=(\r_{i}(t))_{i=1}^{n}:=(\P(J(t)=i))_{i=1}^{n}\in\R^{n},
\end{align*}
satisfies
\begin{align}\label{rode}
\frac{\dd}{\dd t}\r(t)
=R(t)\r(t),\quad t>0,
\end{align}
for some function $R(t):[0,\infty)\mapsto\R^{n\times n}$, and 
\begin{align}\label{sop}
\r(t)
=\Psi(t)\r(0),\quad t\ge0,
\end{align}
where $\Psi(t):(-\infty,\infty)\mapsto\R^{n\times n}$ satisfies
\begin{align}\label{invert}
\Psi(t)\Psi(-t)
=\textup{id},\quad t\in(-\infty,\infty),
\end{align}
where $\textup{id}$ is the identity operator.

Assume $\{X(t)\}_{t\ge0}$ is a stochastic process taking values in the closure of the open set $\domain\subseteq\R^{d}$ whose probability density,
\begin{align*}
q(x,t)\,\dd x
=\P(X(t)=x),
\end{align*}
satisfies
\begin{align}\label{qeqn}
\frac{\partial}{\partial t}q=\DD\L q,\quad x\in\domain\subseteq\R^{d},\,t>0,
\end{align}
where $\L$ and $\DD$ are linear operators satisfying \eqref{linear}.

If $X$ and $J$ are independent, then the joint probability density $\p(x,t)=(\p_{i}(x,t))_{i=1}^{n}$,
\begin{align*}
\p_{i}(x,t)\,\dd x
&=\P(X(t)=x,\, J(t)=i),
\end{align*}
satisfies
\begin{align}\label{evo}
\frac{\partial}{\partial t}\p
=\Psi(t)\DD\Big[\Psi(-t)\L \p\Big]
+{{R}}(t) \p,\quad x\in \domain,\,t>0.
\end{align}
\end{theorem}

\begin{proof}[Proof of Theorem~\ref{thmind}]
Since $X$ and $J$ are independent, the joint probability density is simply the product,
\begin{align*}
\p(x,t)=q(x,t)\r(t),
\end{align*}
and the proof then follows exactly as in \eqref{calc}-\eqref{calc2} with $e^{\pm Rt}$ replaced by $\Psi(\pm t)$.
\end{proof}

Theorem~\ref{thmind} states that if $X$ and $J$ are independent, then their joint density $\p(x,t)$ satisfies the evolution equations in \eqref{evo}. To investigate the converse of Theorem~\ref{thmind}, assume that the joint density $\p(x,t)$ of $X$ and $J$ satisfies the evolution equations in \eqref{evo}. Now, notice that the product $q(x,t)\r(t)$ also satisfies \eqref{evo} if $q(x,t)$ satisfies \eqref{qeqn} and $\r(t)$ satisfies \eqref{rode}. Therefore, if (i) $\p(x,t)$ and $q(x,t)\r(t)$ satisfy the same initial conditions and boundary conditions (or growth conditions if the domain $\domain$ is unbounded) and if (ii) the solution to equation \eqref{evo} with these initial/boundary conditions is unique, then $\p(x,t)=q(x,t)\r(t)$. Therefore, $X$ and $J$ must be independent. In conclusion, the joint density of $X$ and $J$ satisfies \eqref{evo} if and only if $X$ and $J$ are independent, as long as dependencies between $X$ and $J$ are not imposed at $t=0$ or on the spatial boundary.

\section{Examples\label{examples}}

In this section, we illustrate Theorem~\ref{thmind} by applying it to some examples of interest.

\subsection{Some previous results}

To get the result \eqref{langlands} of Langlands et al.\ in \cite{langlands2008}, then we apply Theorem~\ref{thmind} with
\begin{align*}
V=\R,\quad
\L=K_{\alpha}\frac{\partial^{2}}{\partial x^{2}},\quad
\DD=\D,\quad
\Psi(t)=e^{Rt},
\end{align*}
where $R=R(t)$ is constant in time and $n<\infty$.

\subsection{Fractional Fokker-Planck equations}

To find the evolution equations for fractional Fokker-Planck equations with linear reactions, we apply Theorem~\ref{thmind} with $\L$ given by the Fokker-Planck operator in \eqref{fpo} and $\DD=\D$.

\subsection{Other memory kernels}

Theorem~\ref{thmind} shows that the form of the evolution equations in \eqref{evo} holds for more general operators than the Riemann-Liouville fractional derivative $\D$. For example, we can take the time operator $\DD$ to be the integro-differential operator \cite{sokolov2006b, magdziarz2009, carnaffan2017},
\begin{align}\label{bec}
\DD \varphi(t)
=\frac{\dd}{\dd t}\int_{0}^{t}M(t-t')\varphi(t')\,\dd t',
\end{align}
where $M(t):[0,\infty)\mapsto\R$ is the so-called memory kernel. 

\subsection{Superdiffusion}

Using a continuous-time random walk argument and properties of Fourier-Laplace transforms, Schmidt et al.\ \cite{schmidt2007} found that the reaction and movement terms are decoupled in the reaction-superdiffusion equations for L\'{e}vy flights with a single irreversible reaction. In this case, the equation describing the movement (without reaction) of a single molecule is
\begin{align*}
\frac{\partial}{\partial t}q
=K_{\mu}\Delta_{x}^{\mu/2}q,
\end{align*}
where $\Delta_{x}^{\mu/2}$ is the Riesz symmetric fractional derivative acting on $x$. Therefore, the decoupling of reaction and movement terms in the reaction-superdiffusion equations follows from Theorem~\ref{thmind} upon taking the spatial operator to be $\L=K_{\mu}\Delta_{x}^{\mu/2}$ and the time operator $\DD$ to be the identity.

\subsection{Time-dependent rates}

Theorem~\ref{thmind} allows the reaction rate matrix to vary in time. {In particular, suppose that the reaction rate matrix is some given function of time $\{R(t)\}_{t\ge0}$ (which does not depend on spatial position).} Starting from the result of Langlands et al.\ in \eqref{langlands} for constant reaction rates, one might conjecture that the evolution equations for {such} time-dependent reaction rates are
\begin{align}\label{nguess}
\frac{\partial}{\partial t}\p
=e^{\int_{0}^{t}R(s)\,\dd s}\DD\Big[e^{-\int_{0}^{t}R(s)\,\dd s}\L \p\Big]
+R(t) \p,
\end{align}
where the integration $\int_{0}^{t}R(s)\,\dd s$ is performed component wise. {Indeed, \eqref{nguess} has been used to model some physical systems involving a single irreversible reaction with a time-dependent rate \cite{abad2012, abad2013}.} However, Theorem~\ref{thmind} {shows that the conjecture in \eqref{nguess} can fail}, since the solution operator {$\Psi(\pm t)$} in \eqref{sop} for the equation \eqref{rode} is {not always} given {by the matrix exponential $e^{\pm\int_{0}^{t}R(s)\,\dd s}$}. 

{In fact, the two-state irreversible reaction,
\begin{align}\label{1to2}
1\overset{\lambda(t)}{\to}2,
\end{align}
is a rare case of time-dependent reaction rates for which \eqref{nguess} holds, since this is one of the few instances of time-dependent reaction rates in which $\Psi(\pm t)=e^{\pm\int_{0}^{t}R(s)\,\dd s}$ (see Appendix A.2.4 in \cite{aalen2008}).  To illustrate, suppose $J(t)\in\{1,2\}$ models \eqref{1to2} for some reaction rate $\lambda(t)$, and thus assume that the distribution $\r(t)\in\R^{2}$ satisfies the nonautonomous linear system of ordinary differential equations in \eqref{rode} with time-dependent reaction rate matrix,
\begin{align*}
R(t)
=\begin{pmatrix}
-\lambda(t) & 0 \\
\lambda(t) & 0
\end{pmatrix}\in\R^{2\times2}.
\end{align*}
In this case, one can check that the solution operator in \eqref{sop} is indeed the matrix exponential, 
\begin{align*}
\Psi(t)
=e^{\int_{0}^{t}R(s)\,\dd s}
=\begin{pmatrix}
e^{-\int_{0}^{t}\lambda(s)\,\dd s} & 0\\
1-e^{-\int_{0}^{t}\lambda(s)\,\dd s} & 1
\end{pmatrix},\quad\text{for }t\ge0,
\end{align*}
and $\Psi(-t)=e^{-\int_{0}^{t}R(s)\,\dd s}$ for $t>0$.

However, if the reaction scheme is more complicated than \eqref{1to2} and the reaction rates depend on time, then typically $\Psi(\pm t)\neq e^{\pm\int_{0}^{ t}R(s)\,\dd s}$, and thus Theorem~\ref{thmind} shows that \eqref{evo} holds rather than \eqref{nguess}. For example, suppose \eqref{1to2} is now reversible,
\begin{align*}
1\Markov{\lambda_{2}(t)}{\lambda_{1}(t)}2,
\end{align*}
and thus $\r(t)\in\R^{2}$ satisfies \eqref{rode} with
\begin{align*}
R(t)
=\begin{pmatrix}
-\lambda_{1}(t) & \lambda_{2}(t) \\
\lambda_{1}(t) & -\lambda_{2}(t)
\end{pmatrix}\in\R^{2\times2}.
\end{align*}
In this case, it is straightforward to check that the solution operator is
\begin{align}\label{op445}
\Psi(t)
=\begin{pmatrix}
\Psi_{11}(t) & 1-\Psi_{22}(t)\\
1-\Psi_{11}(t) & \Psi_{22}(t)
\end{pmatrix},\quad t\ge0,
\end{align}
where for $i\in\{1,2\}$ and $t\ge0$,
\begin{align*}
\Psi_{ii}(t)
&=e^{-\int_{0}^{t}(\lambda_{1}(s)+\lambda_{2}(s))\,\dd s}\\
&\quad\times\Big(1+\int_{0}^{t}\lambda_{1-i}(s)e^{\int_{0}^{s}(\lambda_{1}(\sigma)+\lambda_{2}(\sigma))\,\dd \sigma}\,\dd s\Big).
\end{align*}
Also, the condition \eqref{invert} implies that the operator evaluated at a negative time argument is the matrix inverse
\begin{align*}
\Psi(-t)
=(\Psi(t))^{-1},\quad\text{for }t>0.
\end{align*}
Note that matrix $\Psi(t)$ is invertible for each $t\ge0$ since it is triangular and the diagonal entries are nonzero. The matrix exponential in this case is
\begin{align}\label{op446}
e^{\int_{0}^{t}R(s)\,\dd s}
=\begin{pmatrix}
1-\chi_{21}(t) & \chi_{12}(t)\\
\chi_{21}(t) & 1-\chi_{12}(t)\\
\end{pmatrix},\quad t\ge0,
\end{align}
where
\begin{align*}
\chi_{ij}(t)
=\frac{\int_0^t \lambda_{j}(s) \, \dd s \left(1-e^{-\int_0^t (\lambda_{1}(s)+\lambda_{2}(s)) \, \dd s}\right)}{\int_0^t \lambda_{1}(s) \, \dd s+\int_0^t \lambda_{2}(s) \, \dd s}.
\end{align*}
It is straightforward to check that \eqref{op445} and \eqref{op446} are generally not equal, except in special cases (such as constant rates, $\lambda_{j}(t)\equiv\lambda_{j}>0$, or equal rates, $\lambda_{1}(t)=\lambda_{2}(t)$ for all $t\ge0$).

Furthermore, it is not merely the presence of a reversible reaction that can cause \eqref{nguess} to fail. For example, suppose $J(t)\in\{1,2,3\}$ has two irreversible reactions,
\begin{align*}
1\overset{\lambda_{1}(t)}{\to}2\overset{\lambda_{2}(t)}{\to}3,
\end{align*}
and its distribution $\r(t)\in\R^{3}$ satisfies \eqref{rode} with 
\begin{align*}
R(t)
=\begin{pmatrix}
-\lambda_{1}(t) & 0 & 0\\
\lambda_{1}(t) & -\lambda_{2}(t) & 0\\
0 & \lambda_{2}(t) & 0
\end{pmatrix}\in\R^{3\times3}.
\end{align*}
The corresponding solution operator for $t\ge0$ is then
\begin{align*}
\Psi(t)
&=\begin{pmatrix}
e^{-\int_{0}^{t}\lambda_{1}(s)\,\dd s} & 0 & 0\\
\Psi_{21}(t) & e^{-\int_{0}^{t}\lambda_{2}(s)\,\dd s} & 0\\
1-e^{-\int_{0}^{t}\lambda_{1}(s)\,\dd s}-\Psi_{21}(t) & 1-e^{-\int_{0}^{t}\lambda_{2}(s)\,\dd s} & 1
\end{pmatrix},
\end{align*}
where
\begin{align*}
\Psi_{21}(t)
&=e^{-\int_{0}^{t}\lambda_{2}(s)\,\dd s}\int_{0}^{t}\lambda_{1}(s)e^{\int_{0}^{s}(\lambda_{2}(\sigma)-\lambda_{1}(\sigma))\,\dd \sigma}\,\dd s.
\end{align*}
For this example, one can check that 
\begin{align*}
\Psi(t)\neq e^{\int_{0}^{t}R(s)\,\dd s},\quad \text{if }t>0,
\end{align*}
except for special cases, and thus \eqref{nguess} is invalid. 

Summarizing, except for a single irreversible reaction, the evolution equation \eqref{nguess} is typically false for time deppendent rates and is corrected by \eqref{evo} in Theorem~\ref{thmind}.

}

\section{Discussion}

We have given a short and elementary proof of the evolution equations for a general class of systems which can combine anomalous motion with linear reaction kinetics. Our results generalize some previous results in \cite{sokolov2006, henry2006, schmidt2007, langlands2008}. The derivations of these previous results employed a variety of mathematical techniques, including continuous-time random walk theory, Fourier and Laplace transforms, Tauberian theorems, and asymptotic expansions. In light of these previous derivations, one might conclude that the form of the evolution equations depends on these finer details. However, we have shown that the evolution equations follow directly from (i) the independence of the stochastic spatial and discrete processes describing a single particle and (ii) the linearity of the integro-differential operators describing particle motion.

Of course, in the present work and in the previous work \cite{sokolov2006, henry2006, schmidt2007, langlands2008}, the evolution equations are not strictly necessary in the sense that the solution to the equations is merely the product of the distributions of the spatial and discrete processes. Nevertheless, these results are expected to be useful for developing models where the independence assumption breaks down. Indeed, evolution equations of a {very} similar form to \eqref{evo} have been derived in \cite{abad2010, yuste2014} for pure subdiffusion with certain space-dependent reaction rates. Furthermore, we agree with Refs.\ \cite{henry2006, langlands2008} that these results could provide a platform for investigating nonlinear reactions, such as those stemming from mass-action kinetics.

{For example, a natural starting place is an irreversible bimolecular reaction of the form \cite{yuste2001prl}
\begin{align*}
A+A\to\varnothing,
\end{align*}
which describes particles that can annihilate each other. However, while the general form of the evolution equations in \eqref{evo} may be instructive for this nonlinear example, it is clear that the approach of the present work cannot be applied directly. Indeed, the present work relied on the independence of the spatial position and discrete state of a single particle. However, it is clear for this example that a single particle is more likely to be in the discrete ``annihilated'' state if it is in a region of space containing a high concentration of particles. Similarly, if we consider a unimolecular reaction of the form
\begin{align*}
A\overset{\lambda(x)}{\to}\varnothing,
\end{align*}
where the first order rate $\lambda(x)>0$ depends on the spatial position $x$ of the particle, it is clear that the particle is more likely to be in the ``annihilated'' state if it is in a region of space where $\lambda(x)$ is large.
}



\begin{acknowledgments}
The author gratefully acknowledges support from the National Science Foundation (DMS-1944574, DMS-1814832, and DMS-1148230).
\end{acknowledgments}


\bibliography{library.bib}

\end{document}